\documentclass[runningheads,a4paper]{llncs}

\usepackage{amssymb, amsfonts}
\usepackage{graphicx}
\usepackage{cite}
\def \x {{\bf x}}
\def \y {{\bf y}}
\def \y {{\bf y}}

\def \c {{\bf c}}
\def \C {{\bf C}}
\def \S {{\bf S}}
\def \u {{\bf u}}
\def \v {{\bf v}}
\def \Z {{\bf Z}}

\def \ZZ {{\mathbb{Z}_8}}

\usepackage{url}
\urldef{\mailsa}\path|{manoj_raut@daiict.ac.in, mankg@daiict.ac.in}|

\begin{document}

\mainmatter  

\title{On Octonary Codes and their Covering Radii}

\titlerunning{On Octonary Codes and their Covering Radii}

%
%
\author{Manoj K. Raut \and Manish K. Gupta} 

%
\authorrunning{Manoj K. Raut \and Manish K. Gupta}

\institute{Dhirubhai Ambani Institute of Information and Communication Technology, Gandhinagar, Gujarat-382007\\
\mailsa\\
\url{http://www.daiict.ac.in}}

%
%

\toctitle{Lecture Notes in Computer Science}
\tocauthor{Authors' Instructions}
\maketitle

\begin{abstract}
This paper introduces new reduction and torsion codes for an octonary code and determines their basic properties. These could be useful for the classification of self-orthogonal and self dual codes over $\ZZ$. We also focus our attention on covering radius problem of octonary codes. In particular, we determine lower and upper bounds of the covering radius of several classes of Repetition codes, Simplex codes of Type $\alpha$ and Type $\beta$ and their duals, MacDonald codes, and Reed-Muller codes over $\mathbb{Z}_8$.
\end{abstract}


\section{Introduction}
Codes over finite rings is an interesting area investigated by many researchers in the last two decades \cite{arzaki, bannai, betsumiya, bonnecaze95, bonnecaze97, chiu, conway, dougherty1, dougherty2, garg, gulliver, gulliver1, hammons1, harada, dougherty3, hammons2, piret, rains, yildiz}. In particular, Octonary codes have received attention by many researchers \cite{arzaki, betsumiya, chiu, dougherty4, garg, gulliver1, piret}. For any octonary linear code, we introduce binary and quaternary (over $\mathbb{Z}_{4}$) reduction and torsion codes and study their basic properties with respect to self-orthogonality and self-duality. One of the other important properties of error correcting codes is covering radius. The covering radius of binary linear codes have been studied in  \cite{cohen1, cohen2}. It is shown in \cite{cohen2, berlekamp} that the problem of computing covering radii of codes is both NP-hard and Co-NP hard. Infact, this problem is strictly harder than any NP-complete problem, unless NP=co-NP. Recently, the covering radius of codes over $\mathbb{Z}_{4}$ has been investigated with respect to Lee and Euclidean distances \cite{aoki}. Several upper and lower bounds on the covering radius of codes has been studied in \cite{aoki}. More recently covering radius of codes over $\mathbb{Z}_{2^s}$ have been defined in \cite{guptadurai14} and upper and lower bounds on the covering radius of several classes of codes over $\mathbb{Z}_{4}$ have been obtained \cite{guptadurai14}. We extend some of these results to octonary codes in this paper.

A {\em linear code} $\mathbb{C},$ of length $n$, over $\mathbb{Z}_{8}$ is an additive subgroup of $\mathbb{Z}_{8}^{n}$.
An element of ${\cal C}$ is called a {\em codeword of} ${\cal C}$ and a {\em generator matrix} of ${\cal C}$ is a
matrix whose rows generate ${\cal C}$.
The {\em Hamming weight} $w_H(\x)$ of a vector $\x$ in $\mathbb{Z}_{8}^n$ 
is the number of non-zero components of $\x$.
The {\em Homogeneous weight} $w_{HW}(\x)$ \cite{consta} of a vector $\x=(x_1,x_2,\ldots,x_n) \in \mathbb{Z}^n_{8}$ 
is given by $\sum_{i=1}^n  w_{HW}(x_i)$ where 
\begin{equation}\label{glw}
w_{HW}(x_i)=\left\{\begin{array}{cc}2,& x_i \neq 4\\
4,& x_i=4.\end{array}\right.
\end{equation}
The {\em Lee weight} $w_L(\x)$ of a vector $\x \in \mathbb{Z}^n_{8}$ is $\sum_{i=1}^n \min\{x_{i}, 8-x_{i}\}$.
The {\em Euclidean weight} $w_E(\x)$ of a vector $\x \in \mathbb{Z}^n_{8}$ is
$\sum_{i=1}^n \min\{x_i^2,(8-x_i)^2\}$. 

The Hamming, Homogeneous, Lee and Euclidean distances $d_H(\x,\y)$, $d_{HW}(\x,\y)$, $d_{L}(\x,\y)$, and $d_E(\x,\y)$ 
between two vectors $\x$ and $\y$ are $w_H(\x-\y)$, $w_{HW}(\x-\y)$, $w_{L}(\x-\y)$ and $w_{E}(\x-\y)$, respectively.
The minimum Hamming, Homogeneous, Lee and Euclidean weights, $d_H, d_{HW}$, $d_L$and $d_E$ of ${\cal C}$ are the smallest Hamming, Homogeneoues, Lee and Euclidean weights among all non-zero codewords of ${\cal C}$ respectively. 
One can define an isometry called the generalized {\em Gray map}  $\phi : \Z_8^{n} \rightarrow \Z_2^{4n}$ as a coordinate-wise extension of the function from $\Z_8$ to $\;\Z_2^{4}$ defined by 
$0 \rightarrow (0,0,0,0), 1 \rightarrow (0,1,0,1), 2 \rightarrow (0,0,1,1), 3 \rightarrow (0,1,1,0), 4 \rightarrow (1,1,1,1), 5 \rightarrow (1,0,1,0), 6 \rightarrow (1,1,0,0), 7 \rightarrow (1,0,0,1)$ \cite{carlet}.
The image $\phi(\C)$, of a linear code $\C$ over $\Z_8$ of length $n$ by the generalized Gray map, is 
a binary code of length $4n$ \cite{hammons1}.

Let $\x=(x_1,x_2,\ldots, x_n)$ and $\y=(y_1,y_2,\ldots,y_n)$ be two vectors in $\mathbb{Z}_{8}^{n}$. Then the inner product of $\x$ and $\y$ is defined by $\x\cdot\y=(x_1y_1+x_2y_2+\ldots+x_ny_n)(~\mbox{mod}~ 8)$. The {\em dual code} ${\cal C}^{\perp}$ of ${\cal C}$ is defined as
$\{ \x \in {\mathbb{Z}}_{8}^n \mid  \x \cdot \y = 0 ~\mbox{for all}~ \y \in {\cal C}\}$ 
where $\x \cdot \y$ is the inner product of $\x$ and $\y$.
${\cal C}$ is {\em self-orthogonal} if ${\cal C} \subseteq {\cal C}^\perp$ and ${\cal C}$ is
{\em self-dual} if ${\cal C}={\cal C}^\perp$.

Two codes are said to be {\em equivalent} if one can be obtained from the
other by permuting the coordinates and (if necessary) changing the signs of
certain coordinates. Codes differing by only a permutation of coordinates are called
{\em permutation-equivalent}. Let ${\cal C}\subseteq \mathbb{Z}_{8}^{n}$. If ${\cal C}$ has $M$ codewords and minimum Homogeneous and Euclidean distances $d_{HW}$ and $d_{E}$ respectively then ${\cal C}$ is called an $(n,M,d_{HW},d_{E})$ code. For more details about octonary codes reader is referred to any of the papers from \cite{gulliver1,dougherty4}.

This paper is organized as follows.
In Section 2, we define new torsion and reduction codes for an octonary code and obtained their basic properties. In Section 3 we present some results of covering radius of octonary codes. Section 4, we discuss about covering radius of Octonary repetition codes. Octonary simplex codes of type $\alpha$ and $\beta$ is discussed in Section 5. In Section 6, we consider MacDonald codes $\mathbb{Z}_{8}$. Finally Section 7 considers Reed-Muller codes and Section 8 considers Octocode. Last section concludes the paper.

\section{Reduction and Torsion Codes}

The standard form of generator matrix $G$ of the linear code ${\cal C}$ over $\mathbb{Z}_{8}$ \cite{dougherty4} is of the form
\begin{equation}\label{general_formZ8}
G=\left(
\begin{array}{cccc}
I_{k_0} & A_{0,1} & A_{0,2} & A_{0,3}  \\
0 & 2I_{k_1} & 2A_{1,2}& 2A_{1,3} \\
0 & 0 & 4I_{k_2} & 4A_{2,3} 
\end{array} \right),
\end{equation} 
where the matrices $A_{i,j}$ are binary matrices for $i > 0$. A code with a generator matrix
in this form is of type $\{k_0, k_1, k_2\}$ and has $8^{k_0}4^{k_1}2^{k_2}$ vectors.

The matrix (\ref{general_formZ8}) can also be written in the following form over $\mathbb{Z}_{8}$:
\begin{equation}\label{gen_Z2}
G=\left(
\begin{array}{cccc}
I_{k_0} & B_{0,1}+2B_{0,1}^{1}+4B_{0,1}^{2} & B_{0,2}+2B_{0,2}^{1}+4B_{0,2}^{2} & B_{0,3}+2B_{0,3}^{1}+4B_{0,3}^{2}  \\
0 & 2I_{k_1} & 2A_{1,2}& 2A_{1,3} \\
0 & 0 & 4I_{k_2} & 4A_{2,3}
\end{array} \right),
\end{equation} 
where $B_{0,i}, B_{0,i}^{1}, \mbox{~and~} B_{0,i}^{2}$ are binary matrices for $i > 0$ and the matrices $A_{i,j}$ are binary matrices for $i > 0$. 
For a quaternary linear code one can define reduction code and torsion code  \cite{conway,hammons1}. These codes have been generalized for a linear code 
${\cal C}$ over $\mathbb{Z}_{8}$ in the form of four binary torsion/reduction codes in \cite{dougherty4}. For $0 \leq i \leq 3$, 
\begin{equation}
Tor_{i}({\cal C})=\{v \pmod 2 \mid 2^{i} v \in {\cal C}\}.
\end{equation} 

The generator matrices of $Tor_{0}({\cal C}),Tor_{1}({\cal C}),Tor_{2}({\cal C})$ are the following three binary matrices:

\begin{equation}
G_{Tor_{0}}=\left(
\begin{array}{cccc}
I_{k_0} & B_{0,1}+2B_{0,1}^{1}+4B_{0,1}^{2} & B_{0,2}+2B_{0,2}^{1}+4B_{0,2}^{2} & B_{0,3}+2B_{0,3}^{1}+4B_{0,3}^{2}  
\end{array} \right)
\end{equation} 

\begin{equation}
G_{Tor_{1}}=\left(
\begin{array}{cccc}
I_{k_0} & B_{0,1}+2B_{0,1}^{1}+4B_{0,1}^{2} & B_{0,2}+2B_{0,2}^{1}+4B_{0,2}^{2} & B_{0,3}+2B_{0,3}^{1}+4B_{0,3}^{2}  \\
0 & I_{k_1} & A_{1,2}& A_{1,3} 
\end{array} \right)
\end{equation} 

\begin{equation}
G_{Tor_{2}}=\left(
\begin{array}{cccc}
I_{k_0} & B_{0,1}+2B_{0,1}^{1}+4B_{0,1}^{2} & B_{0,2}+2B_{0,2}^{1}+4B_{0,2}^{2} & B_{0,3}+2B_{0,3}^{1}+4B_{0,3}^{2}  \\
0 & I_{k_1} & A_{1,2}& A_{1,3} \\
0 & 0 & I_{k_2} & A_{2,3}
\end{array} \right)
\end{equation} 

where 
\begin{equation}
\mid{\cal C}\mid=\mid Tor_{0}({\cal C})\mid\mid Tor_{1}({\cal C})\mid\mid Tor_{2}({\cal C})\mid=2^{3k_0+2k_1+k_2}
\end{equation}

The reduction and torsion code of quaternary linear code can also be generalized for linear codes over $\mathbb{Z}_{8}$ in another interesting way. 
We define two binary (over $\mathbb{Z}_{2}$) torsion codes and two quaternary ($\mathbb{Z}_{4}$) torsion codes for a given linear code over $\mathbb{Z}_{8}$ as follows:

\begin{equation}
\begin{array}{ccl}
{\cal C}^{(1)}& = & \{c\pmod 2\mid c\in {\cal C}\},\\
{\cal C}^{(2)}& = & \{c\pmod 4\mid c\in {\cal C}\},\\
{\cal C}^{(3)}& = & \{c\mid 2c\in {\cal C}\},\\
{\cal C}^{(4)}& = & \{c\mid 4c\in {\cal C}\}.
\end{array}
\end{equation}

The generator matrices $G^{(1)}, G^{(2)}, G^{(3)}, G^{(4)}$ of ${\cal C}^{(1)}, {\cal C}^{(2)}, {\cal C}^{(3)}, {\cal C}^{(4)}$ are obtained from equation (\ref{gen_Z2}) as follows:

\begin{equation}
G^{(1)}=\left(
\begin{array}{cccc}
I_{k_0} & B_{0,1} & B_{0,2} & B_{0,3} 
\end{array} \right),
\end{equation} 

\begin{equation}\label{gen_c2}
G^{(2)}=\left(
\begin{array}{cccc}
I_{k_0} & B_{0,1}+2B_{0,1}^{1} & B_{0,2}+2B_{0,2}^{1} & B_{0,3}+2B_{0,3}^{1}  \\
0 & 2I_{k_1} & 2A_{1,2}& 2A_{1,3} 
\end{array} \right),
\end{equation} 

\begin{equation}\label{gen_c3}
G^{(3)}=\left(
\begin{array}{cccc}
I_{k_0} & B_{0,1}+2B_{0,1}^{1} & B_{0,2}+2B_{0,2}^{1} & B_{0,3}+2B_{0,3}^{1}  \\
0 & I_{k_1} & A_{1,2}& A_{1,3} \\
0 & 0 & 2I_{k_2} & 2A_{2,3}
\end{array} \right),
\end{equation} 

\begin{equation}
G^{(4)}=\left(
\begin{array}{cccc}
I_{k_0} & B_{0,1} & B_{0,2} & B_{0,3} \\
0 & I_{k_1} & A_{1,2} & A_{1,3}\\
0 & 0 & I_{k_2} & A_{2,3}
\end{array} \right).
\end{equation} 

Note that the number of elements of ${\cal C}$ is $8^{k_0}4^{k_1}2^{k_2}=2^{3k_0+2k_1+k_2}$. The number of elements of ${\cal C}^{(1)}, {\cal C}^{(2)}, {\cal C}^{(3)}, {\cal C}^{(4)}$ are $2^{k_0}, 4^{k_0}2^{k_1}, 4^{k_0+k_1}2^{k_2}\;\mbox{and}\; 2^{k_0+k_1+k_2}$ respectively. Thus for a linear code ${\cal C}$ over $\mathbb{Z}_{8}$ we have the following relationship.

\begin{proposition}
\[
 \mid{\cal C}\mid=\mid{\cal C}^{(1)}\mid\times \mid{\cal C}^{(3)}\mid=\mid{\cal C}^{(2)}\mid\times \mid{\cal C}^{(4)}\mid.
\]
\end{proposition}

Note that if $k_1=k_2=0$ then ${\cal C}^{(2)}={\cal C}^{(3)}$. It is easy to observe the following.

\begin{proposition}
\[
{\mathcal C}^{(1)}\subseteq {\cal C}^{(4)}\;\mbox{and}\; {\mathcal C}^{(2)}\subseteq {\cal C}^{(3)}.
\]
\end{proposition}
Next result is a simple generalization of self-orthogonality characterization from \cite{ggg07}. For $\c \in {\cal C}$ and $0 \leq i \leq 7$, let $w_i(\c)$ 
denotes the composition of symbol $i$ in the codeword $\c$.
\begin{proposition}\label{self_orthogonal_cond}
A linear code ${\cal C}$ over $\mathbb{Z}_8$ is self-orthogonal iff each generator matrix of ${\cal C}$ has all its rows $\omega_1+\omega_3+\omega_5+\omega_7+4\omega_2+ 4\omega_6=0 \pmod 8$ and every pair of rows of the generator matrix is orthogonal.
\end{proposition}

\begin{proof}
Proof is straightforward. \qed
\end{proof}

Now we determine few relationships among various reduction and torsion codes if code ${\cal C}$ is self-orthogonal or self-dual.

\begin{proposition}
If ${\cal C}$ is a self-orthogonal code over $\mathbb{Z}_{8}$ then ${\cal C}^{(1)}, {\cal C}^{(4)}$ are self-orthogonal codes over $\mathbb{Z}_{2}$ and  ${\cal C}^{(2)}, {\cal C}^{(3)}$ are self-orthogonal codes over $\mathbb{Z}_{4}$.
\end{proposition}

\begin{proof}
The self-orthogonality of ${\cal C}^{(1)}, {\cal C}^{(2)}$ follows from \cite{dougherty4}. It remains to see the self- orthogonality of ${\cal C}^{(3)}$, and ${\cal C}^{(4)}$. Let $v\in {\cal C}^{(3)}$. By definition of ${\cal C}^{(3)}$ we have $2v\in {\cal C}$. As ${\cal C}$ is self orthogonal, $<2v, u>=0 \pmod 4$ for all $u\in{\cal C}$. So $2\sum v_{i}u_{i}\equiv 0\pmod 8$ for all $u\in{\cal C}$. Then $\sum v_{i}u_{i}\equiv 0\pmod 4$ for all $u\in{\cal C}$. This implies $<v,u>=0$ for all $u\in{\cal C}^{(3)}$ as ${\cal C}^{(3)}$ is a code over $\mathbb{Z}_{4}$. So $v\in {{\cal C}^{(3)}}^{\perp}$. Hence ${\cal C}^{(3)}$ is self orthogonal. The self-orthogonality of ${\cal C}^{(4)}$ can be proved similarly. \qed
\end{proof}

\begin{proposition}\label{1perp2perp34}
If ${\cal C}$ is self orthogonal code over $\mathbb{Z}_{8}$ then
 ${\cal C}^{(4)}\subseteq {{\cal C}^{(1)}}^{\perp}$ and ${\cal C}^{(3)}\subseteq {{\cal C}^{(2)}}^{\perp}$.
\end{proposition}

\begin{proof}
Let $v \in {\cal C}^{(4)}$. By definition of ${\cal C}^{(4)}$, we have $4v \in {\cal C}$. As ${\cal C}$ is self-orthogonal, $<4v,u>=0 \pmod 8$ for all $u\in{\cal C}$. So $4\sum v_{i}u_{i}=0\pmod 8$ for all $u\in{\cal C}$. $\sum v_{i}u_{i}=0\pmod 2$ for all $u\in{\cal C}$. $\sum (v_{i}\pmod 2) (u_{i}\pmod 2)=0\pmod 2$ for all $u \pmod 2 \in{\cal C}$. As ${\cal C}^{(4)}$ is a code over $\mathbb{Z}_{2}$, $v_{i}\pmod 2=v_{i}$ and as ${\cal C}^{(1)}$ is a code over $\mathbb{Z}_{2}$, $\sum v_{i} (u_{i}\pmod 2)=0\pmod 2$ for all $u\in{\cal C}^{(1)}$. So $<v,u>=0\pmod 2$ for all $u\in{\cal C}^{(1)}$. This implies $v \in {{\cal C}^{(1)}}^{\perp}$. The second inclusion ${\cal C}^{(3)}\subseteq {{\cal C}^{(2)}}^{\perp}$ can be proved similarly. \qed
\end{proof}

\begin{proposition}
If ${\cal C}$ is self-dual over $\mathbb{Z}_{8}$ then 
${\cal C}^{(4)}={{\cal C}^{(1)}}^{\perp}$ and ${\cal C}^{(3)}={{\cal C}^{(2)}}^{\perp}$.
\end{proposition}

\begin{proof}
We know that ${\cal C}^{(4)}\subseteq {{\cal C}^{(1)}}^{\perp}$ from Proposition \ref{1perp2perp34}. It remains to show ${{\cal C}^{(1)}}^{\perp}\subseteq {\cal C}^{(4)}$. Let $v\in {{\cal C}^{(1)}}^{\perp}$. So $<v,w>\equiv 0 \pmod 2$ for all $w\in {\cal C}^{(1)}$. $\sum v_{i}w_{i}\equiv 0 \pmod 2$ for all $w\in {\cal C}^{(1)}$. $4\sum v_{i}w_{i}\equiv 0 \pmod 8$ for all $w\in {\cal C}$.  $\sum 4v_{i}w_{i}\equiv 0 \pmod 8$ for all $w\in {\cal C}$. $<4v,w>=0 \pmod 8$ for all $w\in {\cal C}$. $4v \in {\cal C}^{\perp}={\cal C}$. This implies $v \in {\cal C}^{(4)}$. Hence proved.
Proof of second result is similar. \qed
\end{proof}

We know that ${\mathcal C}^{(2)}$ and ${\mathcal C}^{(3)}$ are codes over $\mathbb{Z}_{4}$. Thus it is natural to consider the torsion and reduction code of ${\mathcal C}^{(2)}$ and ${\mathcal C}^{(3)}$. We get the following:

\begin{equation}
\begin{array}{ccl}
{\cal C}^{(21)}&=&\{c\pmod 2\mid c\in {\cal C}^{(2)}\},\\
{\cal C}^{(22)}&=&\{c\mid 2c\in {\cal C}^{(2)}\},\\
{\cal C}^{(31)}&=&\{c\pmod 2\mid c\in {\cal C}^{(3)}\},\\
{\cal C}^{(32)}&=&\{c\mid 2c\in {\cal C}^{(3)}\}.
\end{array}
\end{equation}
The generator matrices $G^{(21)}, G^{(22)}, G^{(31)}, G^{(32)}$ of ${\cal C}^{(21)}, {\cal C}^{(22)}, {\cal C}^{(31)}, {\cal C}^{(32)}$ are obtained from (\ref{gen_c2}) and (\ref{gen_c3}) as follows:
\begin{equation}
G^{(21)}=\left(
\begin{array}{cccc}
I_{k_0} & B_{0,1} & B_{0,2} & B_{0,3} 
\end{array} \right)=G^{(1)},
\end{equation} 

\begin{equation}
G^{(22)}=\left(
\begin{array}{cccc}
I_{k_0} & B_{0,1} & B_{0,2} & B_{0,3}  \\
0 & I_{k_1} & A_{1,2}& A_{1,3}
\end{array} \right),
\end{equation} 

\begin{equation}
G^{(31)}=\left(
\begin{array}{cccc}
I_{k_0} & B_{0,1} & B_{0,2} & B_{0,3}  \\
0 & I_{k_1} & A_{1,2}& A_{1,3}
\end{array} \right)=G^{(22)},
\end{equation} 

\begin{equation}
G^{(32)}=\left(
\begin{array}{cccc}
I_{k_0} & B_{0,1} & B_{0,2} & B_{0,3}  \\
0 & I_{k_1} & A_{1,2}& A_{1,3} \\
0 & 0 & I_{k_2} & A_{2,3}
\end{array} \right).
\end{equation} 

All the codes ${\cal C}^{(21)}, {\cal C}^{(22)}, {\cal C}^{(31)}, {\cal C}^{(32)}$ are codes over $\mathbb{Z}_{2}$. It is easy to see the following results from their generator matrices:

\begin{proposition}\label{subsubcode1}
If ${\cal C}$ is a code over $\mathbb{Z}_{8}$ then ${\cal C}^{(21)}={\cal C}^{(1)}, {\cal C}^{(1)}\subseteq {\cal C}^{(22)}, {\cal C}^{(31)}={\cal C}^{(22)}, {\cal C}^{(31)} \subseteq {\cal C}^{(32)}, \mbox{~and~} {\cal C}^{(32)}={\cal C}^{(4)}.$
\end{proposition}

Also it is natural to obtain the following from \cite{conway}.

\begin{proposition}
If ${\cal C}$ is self-orthogonal over $\mathbb{Z}_{8}$ then
\begin{enumerate}
\item  ${\cal C}^{(21)}\subseteq {\cal C}^{(22)}\subseteq {{\cal C}^{(21)}}^{\perp}$.
\item  ${\cal C}^{(31)}\subseteq {\cal C}^{(32)}\subseteq {{\cal C}^{(31)}}^{\perp}$.
\end{enumerate}
\end{proposition}

Thus we have interesting family of codes from a linear octonary codes having beautiful inclusions. 





\section{Covering Radius of Octonary Codes}
In this section, first we collect some known facts of the covering radius of codes over $\mathbb{Z}_{8}$ with respect to Homogeneous and Euclidean distances \cite{guptadurai14} and then derive some of its properties. Let $d$ be either a Homogeneous distance or Euclidean distance. Then the covering radius of code ${\cal C}$ over $\mathbb{Z}_{8}$ with respect to distance $d$ is given by 
$$r_d({\cal C})= \max_{\u \in \mathbb{Z}_{8}^{n}}\left\{\min_{\c \in {\cal C}}d(\u,\c)\right\}.$$    
We can easily see \cite{guptadurai14} that $r_d({\cal C})$ is the minimum value $r_d$ such that 
$$\mathbb{Z}_{8}^{n}= \cup_{\c \in {\cal C}} S_{r_d}(\c)$$ where
$$S_{r_d}(\u)=\left\{\v \in \mathbb{Z}_{8}^{n} \mid d(\u,\v) \leq r_d \right\}$$ for 
any element $\u \in \mathbb{Z}_{8}^{n}$.

The coset of ${\cal C}$ is the translate $\u + {\cal C} = \left\{\u + \c \mid \c \in {\cal C} \right\}$ where $\u\in \Z_{8}^{n}$. A vector of least
weight in a coset is called a {\em coset leader}. The following proposition is well known \cite{guptadurai14}. 

\begin{proposition}
The covering radius of ${\cal C}$ with respect to the general distance $d$ is the 
largest minimum weight among all cosets.
\end{proposition}

\begin{proposition}\label{HW_E_L}
For any octonary code over $\mathbb{Z}_{8}$, 
\begin{equation}
\begin{array}{ccccc}
\frac{1}{2}r_{HW}({\cal C})&\leq & r_{E}({\cal C})&\leq& 5r_{HW}({\cal C}),\\
r_{L}({\cal C})&\leq& r_{E}({\cal C}),&&\\
r_{{HW}}({\cal C})&\leq& 2r_{L}({\cal C}).&&
\end{array}
\end{equation}
\end{proposition}
\begin{proof}
We observe that $\frac{1}{2}d_{HW}(\x,\y)\leq d_{E}(\x,\y)\leq 5d_{HW}(\x,\y)$, so the first inequality follows. As  $d_{L}(\x,\y)\leq d_{E}(\x,\y)$, the second inequality follows. Further the third inequality follows since we have  $d_{HW}(\x,\y)\leq 2d_{L}(\x,\y)$. \qed
\end{proof}

The following proposition is also well known \cite{guptadurai14}.

\begin{proposition}\label{HW_H}
Let ${\cal C}$ be a code over $\mathbb{Z}_{8}$ and $\phi({\cal C})$ the generalized Gray map image of ${\cal C}$.
Then $r_{HW}({\cal C}) = r_{H}(\phi({\cal C}))$.
\end{proposition}

The following two results are two upper bounds of the covering radius of codes over $\mathbb{Z}_{8}$ with respect to Homogeneous weight. 

\begin{proposition}{\bf(Sphere-Covering Bound)}\label{sphere}
For any code ${\cal C}$ of length $n$ over $\mathbb{Z}_{8}$,

\begin{equation}
\begin{array}{rcl} 
\frac{2^{4n}}{|{\cal C}|} &\leq& \sum_{i=0}^{r_{HW}({\cal C})} {4n  \choose i},\\
\frac{2^{4n}}{|{\cal C}|} &\leq& \sum_{i=0}^{r_{E}({\cal C})}V_{i},\\
where~ \sum_{i=0}^{16n} V_{i}x^{i}&=&(1+2x+2x^{4}+2x^{9}+x^{16})^{n}.
\end{array}
\end{equation}
\end{proposition}

\begin{proof}
The proof of both inequality over $\mathbb{Z}_{8}$ is similar to the proof over $\mathbb{Z}_{4}$ given in \cite{aoki} (see also   \cite{guptadurai14}) and hence omitted. \qed
\end{proof}

Let ${\cal C}$ be a code over $\mathbb{Z}_{8}$ and let 
$s({\cal C}^{\perp})=|\left\{i \mid A_{i}({\cal C}^{\perp}) \neq 0, i \neq 0 \right\}|$, 
where $A_i({\cal C}^{\perp})$ is the number of codewords of homogenous weight $i$ 
in ${\cal C}^{\perp}$. 

\begin{theorem}{\bf (Delsarte Bound)}\label{delsarte}
Let ${\cal C}$ be a code over $\mathbb{Z}_{8}$ then $r_{HW}({\cal C}) \leq s({\cal C}^{\perp})$ and $r_{E}({\cal C})\leq 5s({\cal C}^{\perp})$.
\end{theorem}
\begin{proof}
First result is obtained in \cite{guptadurai14}. 
Second result follows from \cite{aoki} and Proposition \ref{HW_E_L}.\qed
\end{proof}

The following result of Mattson \cite{cohen1} is useful  for computing covering radii of codes over rings \cite{guptadurai14}. 

\begin{proposition} {\bf (Mattson)}\label{mattson}
If ${\cal C}_0$ and ${\cal C}_1$ are codes over $\mathbb{Z}_{8}$ generated by matrices $G_0$ and $G_1$ respectively and 
if ${\cal C}$ is the code generated by 
\[
G =  \left( \begin{array}{c|c} 
0 & G_1 \\\hline
G_0 & A 
\end{array}\right),
\]
then $r_d({\cal C}) \leq r_d({\cal C}_0) + r_d({\cal C}_1)$ and  the covering radius of ${\cal D}$ (concatenation of ${\cal C}_0$ and ${\cal C}_1$) 
satisfies the following 
\[
r_d({\cal D}) \geq r_d({\cal C}_0) + r_d({\cal C}_1), 
\]
for all distances $d$ over $\mathbb{Z}_{8}$.
\end{proposition}

Now we determine a bound on the covering radius of octonary code and its corresponding reduction and torsion codes. The following result is a generalization of Theorem 4.4 of \cite{aoki}.
\begin{theorem}
For a code over $\mathbb{Z}_{8}$, let $d_1, d_2, d_3, d_4$ denote the minimum Hamming distances of linear codes ${\cal C}^{(1)}, {\cal C}^{(2)},{\cal C}^{(3)},{\cal C}^{(4)}$ respectively. If $d_1\geq 8, d_2\geq 18, d_3\geq \frac{25}{4}, d_4\geq \frac{25}{16}$ then

\[
\begin{array}{ccc}
r_{E}(C)&\geq& 9\mbox{~min~}\Big\{\big\lfloor\frac{d_1}{8}\big\rfloor, \big\lfloor\frac{d_2}{18}\big\rfloor, 4\big\lfloor\frac{d_3}{25}\big\rfloor, 16\big\lfloor\frac{d_4}{25}\big\rfloor\Big\},\\
r_{HW}(C)&\geq& 2\mbox{~min~}\Big\{\big\lfloor\frac{d_1}{8}\big\rfloor, \big\lfloor\frac{d_2}{18}\big\rfloor, 4\big\lfloor\frac{d_3}{25}\big\rfloor, 16\big\lfloor\frac{d_4}{25}\big\rfloor\Big\}.
\end{array} 
\]
\end{theorem}
\begin{proof}
Let $t=\mbox{~min~}\Big\{\big\lfloor\frac{d_1}{8}\big\rfloor, \big\lfloor\frac{d_2}{18}\big\rfloor, 4\big\lfloor\frac{d_3}{25}\big\rfloor, 16\big\lfloor\frac{d_4}{25}\big\rfloor\Big\}$. Hence $t>0$. Let $\x =\\(00\ldots 0\underbrace{44\ldots 4}_{t})$. Let ${\cal C}$ be a code over $\mathbb{Z}_{8}$. Let $\c =(c_1,c_2,\ldots,c_n) \in{\cal C}$ such that $c_i=0 \mbox{~ or~ } 4$. Hence $\frac{\c}{4}\in {\cal C}^{(4)}$. So $wt({\cal C}^{(4)})\geq d_4$ as the minimum Hamming distance of ${{\cal C}}^{(4)}$ is $d_4$. Thus $wt(\frac{\c}{4})\geq d_4\geq t$. Let $\c=(00\ldots 0\underbrace{44\ldots 4}_{\geq d_4})$. Hence

\[
\begin{array}{ccccc}
d_{E}(\c,\x)&=&16(d_4-t)&\geq& 9t,\\
d_{HW}(\c,\x)&=&4(d_4-t)&\geq& 2t.
\end{array}
\]
Similarly for $\c \in{\cal C}$ such that $\frac{\c}{2}\in {\cal C}^{(3)}$ we get
\[
\begin{array}{ccccc}
d_{E}(\c,\x)&\geq&4d_3-16t&\geq& 9t,\\
d_{HW}(\c,\x)&\geq&2d_3-4t&\geq& 2t.
\end{array}
\]
For  $\c\in{\cal C}$ such that $\c \pmod 4 \in{\cal C}^{(2)}$ we have
\[
\begin{array}{ccccc}
d_{E}(\c,\x)&\geq&d_2-9t&\geq& 9t,\\
d_{HW}(\c,\x)&=&2d_2&\geq& 2t.
\end{array}
\]
Finally for $\c \in{\cal C}$ such that $\c \pmod 2 \in{\cal C}^{(1)}$ we have
 \[
\begin{array}{ccccc}
d_{E}(\c,\x)&=&d_1+8t&\geq& 9t,\\
d_{HW}(\c,\x)&=&2d_1&\geq& 2t.
\end{array}
\]
Hence the result follows. \qed
\end{proof}

\section{Octonary Repetition Codes}

A $q$-ary repetition code ${\cal C}$ over a finite field $\mathbb{F}_q=\{\alpha_0=0, \alpha_1=1,\alpha_2,\alpha_3,\ldots, \alpha_{q-2}\}$ is an [n,1,n]-code ${\cal C}=\{\overline{\alpha}\mid\alpha\in\mathbb{F}_q\}$, where $\overline{\alpha} =\{\alpha,\alpha,\ldots,\alpha\}$. The covering radius of ${\cal C}$ is $\lceil\frac{n(q-1)}{q}\rceil$\cite{durairajan}. In \cite{guptadurai14}, several classes of repetition codes over $\mathbb{Z}_4$ have been studied and their covering radius has been obtained. Now we generalize those results for codes over $\mathbb{Z}_8$.  Consider the repetition codes over $\mathbb{Z}_8$. One can define seven basic repetition codes ${\cal C}_{\alpha_i}$, ($1\leq i\leq n$) of length $n$ over $\mathbb{Z}_8$ generated by $G_{\alpha_1}=[\underbrace{11\ldots 1}_{n}]$, $G_{\alpha_2}=[\underbrace{22 \ldots 2}_{n}]$, $G_{\alpha_3}=[\underbrace{33 \ldots 3}_{n}]$, $G_{\alpha_4}=[\underbrace{44 \ldots 4}_{n}]$, $G_{\alpha_5}=[\underbrace{55 \ldots 5}_{n}]$, $G_{\alpha_6}=[\underbrace{66 \ldots 6}_{n}]$, $G_{\alpha_7}=[\underbrace{77\ldots 7}_{n}]$. So the repetition codes are ${\cal C}_{\alpha_1} = {\cal C}_{\alpha_3} = {\cal C}_{\alpha_5} = {\cal C}_{\alpha_7}=\{(00\ldots,0),  (11\ldots1), (22\ldots 2), (33\ldots 3), (44 \ldots 4), (55 \ldots 5), \\(66 \ldots 6), (77 \ldots7)\},$ 
${\cal C}_{\alpha_2}={\cal C}_{\alpha_6}=\{(00\ldots 0),(22\ldots 2), (44 \ldots 4),(66 \ldots 6)\}$ and ${\cal C}_{\alpha_4}=\{(00 \ldots 0), (44 \ldots 4)\}$.  The following theorems determine the covering radius of ${\cal C}_{\alpha_i}$ for $1\leq i\leq 7$.

\begin{theorem}\label{1357}
$r_E({\cal C}_{\alpha_1})=r_E({\cal C}_{\alpha_3})=r_E({\cal C}_{\alpha_5})=r_E({\cal C}_{\alpha_7})=\frac{11n}{2}$ ~and~  $r_{HW}({\cal C}_{\alpha_1})\\=r_{HW}({\cal C}_{\alpha_3})=r_{HW}({\cal C}_{\alpha_5})=r_{HW}({\cal C}_{\alpha_7})= 2n$
\end{theorem}
\begin{proof}
We know that $r_{E}({\cal C}_{\alpha_i})=\mbox{max}_{x\in\mathbb{Z}_{8}^{n}}\{d_{E}(x,{\cal C}_{\alpha_i})\}$. Let $\x\in \mathbb{Z}_{8}^{n}$. If $\x$ has composition $(\omega_0,\omega_1,\omega_2,\omega_3,\omega_4,\omega_5,\omega_6,\omega_7)$ where $\sum_{i=0}^{7} \omega_i = n$, then $d_{E}(\x,\bar{0})=n-\omega_0+3\omega_2+8\omega_3+15\omega_4+8\omega_5+3\omega_6$, $d_{E}(\x,\bar{1})=n-\omega_1+3\omega_3+8\omega_4+15\omega_5+8\omega_6+3\omega_7$,  $d_{E}(\x,\bar{2})=n-\omega_2+3\omega_0+3\omega_4+8\omega_5+15\omega_6+8\omega_7$, $d_{E}(\x,\bar{3})=n-\omega_3+8\omega_0+3\omega_1+3\omega_5+8\omega_6+15\omega_7$, $d_{E}(\x,\bar{4})=n-\omega_4+15\omega_0+8\omega_1+3\omega_2+3\omega_6+8\omega_7$, $d_{E}(\x,\bar{5})=n-\omega_5+8\omega_0+15\omega_1+8\omega_2+3\omega_3+3\omega_7$, $d_{E}(\x,\bar{6})=n-\omega_6+3\omega_0+8\omega_1+15\omega_2+8\omega_3+3\omega_4$, and $d_{E}(\x,\bar{7})=n-\omega_7+3\omega_1+8\omega_2+15\omega_3+8\omega_4+3\omega_5$. Then 
\[
\begin{array}{ccl}
d_{E}(\x,{\cal C}_{\alpha_1})&\leq&\frac{8n+36(\omega_0+\omega_1+\omega_2+\omega_3+\omega_4+\omega_5+\omega_6+\omega_7)}{8}=\frac{11n}{2}.\\
\end{array}
\]
Thus $r_{E}({\cal C}_{\alpha_1})\leq\frac{11n}{2}$. 

Let $\x=\underbrace{00\ldots 0}_{t}\underbrace{11\ldots 1}_{t}\underbrace{22\ldots 2}_{t}\underbrace{33\ldots 3}_{t}\underbrace{44\ldots 4}_{t}\underbrace{55\ldots 5}_{t}\underbrace{66\ldots 6}_{t}\underbrace{77\ldots 7}_{n-7t}\in\mathbb{Z}_{8}^{n}$, where $t=\lfloor\frac{n}{8}\rfloor$. Then $d_{E}(\x,\bar{0})=n+36t$, $d_{E}(\x,\bar{1})=4n+12t$, $d_{E}(\x,\bar{2})=9n-28t$, $d_{E}(\x,\bar{3})=16n-84t$, $d_{E}(\x,\bar{4})=9n-28t$, $d_{E}(\x,\bar{5})=4n+12t$, $d_{E}(\x,\bar{6})=n+36t$, $d_{E}(\x,\bar{7})=44t$. Thus
\[
\begin{array}{ccl} 
r_{E}({\cal C}_{\alpha_1}) &\geq&\frac{44n+36t+12t-28t-84t-28t+12t+36t+44t}{8}=  \frac{11n}{2}.\\
\end{array}
 \]
Thus  $r_E({\cal C}_{\alpha_1}) = r_E({\cal C}_{\alpha_3}) = r_E({\cal C}_{\alpha_5}) = r_E({\cal C}_{\alpha_7})=\frac{11n}{2}$.
The gray map $\phi({\cal C}_{\alpha_1})$ will be a binary repetition code of length $4n$. Thus $r_{HW}({\cal C}_{\alpha_1})=\lceil\frac{4n(2-1)}{2}\rceil=2n= r_{HW}({\cal C}_{\alpha_3})=r_{HW}({\cal C}_{\alpha_5})=r_{HW}({\cal C}_{\alpha_7})$.\qed
\end{proof}

\begin{theorem}\label{26}
$r_E({\cal C}_{\alpha_2})=r_E({\cal C}_{\alpha_6})= 6n$ and $r_{HW}({\cal C}_{\alpha_2})=r_{HW}({\cal C}_{\alpha_6})=2n$
\end{theorem}

\begin{proof}
The proof is similar to proof of Theorem \ref{1357}, hence omitted.

\end{proof}

\begin{theorem}\label{4}
$r_E({\cal C}_{\alpha_4})=8n$ and $r_{HW}({\cal C}_{\alpha_4})=2n$
\end{theorem}

\begin{proof} The proof is similar to proof of Theorem \ref{1357}, hence omitted.



\end{proof}

In order to determine the covering radius of Simplex code $S_{k}^{\alpha}$ over $\mathbb{Z}_{8}$, we have to define a block repetition code over $\mathbb{Z}_{8}$ and find its covering radius. Thus the covering radius of the block repetition code $BRep^{m_1+m_2+\ldots+m_7}$:$(n=m_1+m_2+\ldots+m_7, M=8,  d_{HW}=\mbox{min}\{2m_1+2m_2+2m_3+4m_4+2m_5+2m_6+2m_7, 2m_1+4m_2+2m_3+2m_5+4m_6+2m_7, 4m_1+4m_3+4m_5+4m_7\}, d_{E}=\mbox{min}\{m_1+4m_2+9m_3+16m_4+9m_5+4m_6+m_7, 4m_1+16m_2+4m_3+4m_5+16m_6+4m_7, 9m_1+4m_2+m_3+16m_4+m_5+4m_6+9m_7, 16m_1+16m_3+16m_5+16m_7\})$ generated by $G=[\underbrace{11\ldots 1}_{m_1}\underbrace{22\ldots2}_{m_2}\underbrace{33\ldots3}_{m_3}\underbrace{44\ldots4}_{m_4}\underbrace{55\ldots5}_{m_5}\underbrace{66\ldots6}_{m_6}\underbrace{77\ldots7}_{m_7}]$ is given in the following theorems.

\begin{theorem}\label{brep_e_1to7}
$r_{E}(BRep^{m_1+m_2+\ldots+m_7})=\frac{11}{2}(m_1+m_3+m_5+m_7)+6(m_2+m_6)+8m_4$.\\ 
\end{theorem}
\begin{proof}
By proposition \ref{mattson} and Theorem \ref{1357}, \ref{26}, \ref{4} we have $r_{E}(BRep^{m_1+m_2+\ldots+m_7})\geq \frac{11}{2}(m_1+m_3+m_5+m_7)+6(m_2+m_6)+8m_4$.

On the other hand, let $\x=(\x_1\mid\x_2\mid\x_3\mid\x_4\mid\x_5\mid\x_6\mid\x_7)\in\mathbb{Z}_{8}^{m_1+m_2+\ldots+m_7}$ with $\x_1,\x_2,\x_3,\x_4,\x_5,\x_6,\x_7$ have compositions $(p_0,p_1,\ldots,p_7)$, $(q_0,q_1,\ldots,q_7)$, $(r_0,r_1,\ldots,r_7)$, $(s_0,s_1,\ldots,s_7)$, $(t_0,t_1,\ldots,t_7)$, $(u_0,u_1,\ldots,u_7)$, $(w_0,w_1,\ldots,w_7)$ such that $p_0+p_1+\ldots+p_7=m_1$, $q_0+q_1+\ldots+q_7=m_2$, $r_0+r_1+\ldots+r_7=m_3$, $s_0+s_1+\ldots+s_7=m_4$, $t_0+t_1+\ldots+t_7=m_5$, $u_0+u_1+\ldots+u_7=m_6$, $w_0+w_1+\ldots+w_7=m_7$. 

$d_{E}(\x,\bar{0})=m_1+m_2+m_3+m_4+m_5+m_6+m_7-p_0+3p_2+8p_3+15p_4+8p_5+3p_6-q_0+3q_2+8q_3+15q_4+8q_5+3q_6-r_0+3r_2+8r_3+15r_4+8r_5+3r_6-s_0+3s_2+8s_3+15s_4+8s_5+3s_6-t_0+3t_2+8t_3+15t_4+8t_5+3t_6-u_0+3u_2+8u_3+15u_4+8u_5+3u_6-w_0+3w_2+8w_3+15w_4+8w_5+3w_6$, where $\bar{0}$ is the first vector of $BRep^{m_1+m_2+\ldots+m_7}$.

$d_{E}(\x,\c_1)=m_1+m_2+m_3+m_4+m_5+m_6+m_7-p_1+3p_3+8p_4+15p_5+8p_6+3p_7-q_2+3q_0+3q_4+8q_5+15q_6+8q_7-r_3+8r_0+3r_1+3r_5+8r_6+15r_7-s_4+15s_0+8s_1+3s_2+3s_6+8s_7-t_5+8t_0+15t_1+8t_2+3t_3+3t_7-u_6+3u_0+8u_1+15u_2+8u_3+3u_4-w_7+3w_1+8w_2+15w_3+8w_4+3w_5$, where $\c_1=(\underbrace{11\ldots 1}_{m_1}\underbrace{22\ldots 2}_{m_2}\underbrace{33\ldots 3}_{m_3}\underbrace{44\ldots 4}_{m_4}\underbrace{55\ldots 5}_{m_5}\underbrace{66\ldots 6}_{m_6}\underbrace{77\ldots 7}_{m_7})$ is the second vector of $BRep^{m_1+m_2+\ldots+m_7}$.

$d_{E}(\x,\c_2)=m_1+m_2+m_3+m_4+m_5+m_6+m_7-p_2+3p_0+3p_4+8p_5+15p_6+8p_7-q_4+15q_0+8q_1+3q_2+3q_6+8q_7-r_6+3r_0+8r_1+15r_2+8r_3+3r_4-s_0+3s_2+8s_3+15s_4+8s_5+3s_6-t_2+3t_0+3t_4+8t_5+15t_6+8t_7-u_4+15u_0+8u_1+3u_2+3u_6+8u_7-w_6+3w_0+8w_1+15w_2+8w_3+3w_4$, where $\c_2=(\underbrace{22\ldots 2}_{m_1}\underbrace{44\ldots 4}_{m_2}\underbrace{66\ldots 6}_{m_3}\underbrace{00\ldots 0}_{m_4}\underbrace{22\ldots 2}_{m_5}\underbrace{44\ldots 4}_{m_6}\underbrace{66\ldots 6}_{m_7})$ is the third vector of $BRep^{m_1+m_2+\ldots+m_7}$.

$d_{E}(\x,\c_3)=m_1+m_2+m_3+m_4+m_5+m_6+m_7-p_3+8p_0+3p_1+3p_5+8p_6+15p_7-q_6+3q_0+8q_1+15q_2+8q_3+3q_4-r_1+3r_3+8r_4+15r_5+8r_6+3r_7-s_4+15s_0+8s_1+3s_2+3s_6+8s_7-t_7+3t_1+8t_2+15t_3+8t_4+3t_5-u_2+3u_0+3u_4+8u_5+15u_6+8u_7-w_5+8w_0+15w_1+8w_2+3w_3+3w_7$, where $\c_3=(\underbrace{33\ldots 3}_{m_1}\underbrace{66\ldots 6}_{m_2}\underbrace{11\ldots 1}_{m_3}\underbrace{44\ldots 4}_{m_4}\underbrace{77\ldots 7}_{m_5}\underbrace{22\ldots 2}_{m_6}\underbrace{55\ldots 5}_{m_7})$ is the fourth vector of $BRep^{m_1+m_2+\ldots+m_7}$.

$d_{E}(\x,\c_4)=m_1+m_2+m_3+m_4+m_5+m_6+m_7-p_4+15p_0+8p_1+3p_2+3p_6+8p_7-q_0+3q_2+8q_3+15q_4+8q_5+3q_6-r_4+15r_0+8r_1+3r_2+3r_6+8r_7-s_0+3s_2+8s_3+15s_4+8s_5+3s_6-t_4+15t_0+8t_1+3t_2+3t_6+8t_7-u_0+3u_2+8u_3+15u_4+8u_5+3u_6-w_4+15w_0+8w_1+3w_2+3w_6+8w_7$, where $\c_4=(\underbrace{44\ldots 4}_{m_1}\underbrace{00\ldots 0}_{m_2}\underbrace{44\ldots 4}_{m_3}\underbrace{00\ldots 0}_{m_4}\underbrace{44\ldots 4}_{m_5}\underbrace{00\ldots 0}_{m_6}\underbrace{44\ldots 4}_{m_7})$ is the fifth vector of $BRep^{m_1+m_2+\ldots+m_7}$.

$d_{E}(\x,\c_5)=m_1+m_2+m_3+m_4+m_5+m_6+m_7-p_5+8p_0+15p_1+8p_2+3p_3+3p_7-q_2+3q_0+3q_4+8q_5+15q_6+8q_7-r_7+3r_1+8r_2+15r_3+8r_4+3r_5-s_4+15s_0+8s_1+3s_2+3s_6+8s_7-t_1+3t_3+8t_4+15t_5+8t_6+3t_7-u_6+3u_0+8u_1+15u_2+8u_3+3u_4-w_3+8w_0+3w_1+3w_5+8w_6+15w_7$, where $\c_5=(\underbrace{55\ldots 5}_{m_1}\underbrace{22\ldots 2}_{m_2}\underbrace{77\ldots 7}_{m_3}\underbrace{44\ldots 4}_{m_4}\underbrace{11\ldots 1}_{m_5}\underbrace{66\ldots 6}_{m_6}\underbrace{33\ldots 3}_{m_7})$ is the sixth vector of $BRep^{m_1+m_2+\ldots+m_7}$.

$d_{E}(\x,\c_6)=m_1+m_2+m_3+m_4+m_5+m_6+m_7-p_6+3p_0+8p_1+15p_2+8p_3+3p_4-q_4+15q_0+8q_1+3q_2+3q_6+8q_7-r_2+3r_0+3r_4+8r_5+15r_6+8r_7-s_0+3s_2+8s_3+15s_4+8s_5+3s_6-t_6+3t_0+8t_1+15t_2+8t_3+3t_4-u_4+15u_0+8u_1+3u_2+3u_6+8u_7-w_2+3w_0+3w_4+8w_5+15w_6+8w_7$, where $\c_6=(\underbrace{66\ldots 6}_{m_1}\underbrace{44\ldots 4}_{m_2}\underbrace{22\ldots 2}_{m_3}\underbrace{00\ldots 0}_{m_4}\underbrace{66\ldots 6}_{m_5}\underbrace{44\ldots 4}_{m_6}\underbrace{22\ldots 2}_{m_7})$ is the seventh vector of $BRep^{m_1+m_2+\ldots+m_7}$.

$d_{E}(\x,\c_7)=m_1+m_2+m_3+m_4+m_5+m_6+m_7-p_7+3p_1+8p_2+15p_3+8p_4+3p_5-q_6+3q_0+8q_1+15q_2+8q_3+3q_4-r_5+8r_0+15r_1+8r_2+3r_3+3r_7-s_4+15s_0+8s_1+3s_2+3s_6+8s_7-t_3+8t_0+3t_1+3t_5+8t_6+15t_7-u_2+3u_0+3u_4+8u_5+15u_6+8u_7-w_1+3w_3+8w_4+15w_5+8w_6+3w_7$, where $\c_7=(\underbrace{77\ldots 7}_{m_1}\underbrace{66\ldots 6}_{m_2}\underbrace{55\ldots 5}_{m_3}\underbrace{44\ldots 4}_{m_4}\underbrace{33\ldots 3}_{m_5}\underbrace{22\ldots 2}_{m_6}\underbrace{11\ldots 1}_{m_7})$ is the eighth vector of $BRep^{m_1+m_2+\ldots+m_7}$.
Thus 
\[
\begin{array}{ccc}
d(\x, BRep^{m_1+m_2+\ldots+m_7}) &\leq& \frac{11}{2}(m_1+m_3+m_5+m_7)+6(m_2+m_6)+8m_4. 
\end{array} 
\]
Hence the equality. 
\end{proof}

\begin{theorem}\label{uboundre}
$\min\{2m_1+2m_2+2m_3+2m_4+2m_5+2m_6+2m_7, 2m_2+2m_3+2m_4+4m_5+2m_6+2m_7, 2m_1+2m_2+2m_4+2m_5+2m_6+4m_7, 4m_1+2m_2+2m_3+2m_4+2m_6+2m_7, 2m_1+2m_2+4m_3+2m_4+2m_5+2m_6\}\leq r_{HW}(BRep^{m_1+m_2+\ldots+m_7})\leq 11(m_1+m_3+m_5+m_7)+12(m_2+m_6)+16m_4$.
\end{theorem}

\begin{proof} By choosing $\x=(\underbrace{11\ldots\ldots\ldots 1}_{m_1+m_2+\ldots+m_7}) \in \ZZ^{m_1+m_2+\ldots+m_7}$ and computing the homogenous distance from each codeword 
we get
$d_{HW}(\x, BRep^{m_1+m_2+\ldots+m_7})= \min\{2m_1+2m_2+2m_3+2m_4+2m_5+2m_6+2m_7, 2m_2+2m_3+2m_4+4m_5+2m_6+2m_7, 2m_1+2m_2+2m_4+2m_5+2m_6+4m_7, 4m_1+2m_2+2m_3+2m_4+2m_6+2m_7, 2m_1+2m_2+4m_3+2m_4+2m_5+2m_6\}$. Hence the first inequality follows. The second inequality follows from Proposition \ref{HW_E_L} and Theorem \ref{brep_e_1to7}. \qed
\end{proof}

\section{Octonary Simplex Codes of Type $\alpha$ and $\beta$}
Simplex codes of type $\alpha$ and $\beta$ have been studied in \cite{gupta}. The linear code $\S_{k}^{\alpha}$ is a type $\alpha$ simplex code over $\mathbb{Z}_{8}$ with parameters $(n=8^{k},M=8^{k}, d_{HW}=2^{3(k+1)-2})$ generated by 
\begin{equation}\label{skalpha}
G_k^{\alpha} = \left[\begin{array}{c|c|c|c|c|c|c|c} 00 \cdots 0 & 11 \cdots 1 &
22 \cdots 2 & 33 \cdots 3 &44 \cdots 4&55 \cdots 5&66 \cdots 6&77 \cdots 7\\\hline
G_{k-1}^{\alpha} & G_{k-1}^{\alpha} & G_{k-1}^{\alpha}&G_{k-1}^{\alpha}& G_{k-1}^{\alpha}& G_{k-1}^{\alpha}& G_{k-1}^{\alpha}& G_{k-1}^{\alpha} 
\end{array}\right]
\end{equation}
with $G_1^{\alpha}=[01234567]$. The number of vectors in $\S_{k}^{\alpha}$ is $2^{3k}$. 
The dual code of $\S_{k}^{\alpha}$ is denoted by ${S_k^{\alpha}}^{\perp}$.

The linear code $\S_{k}^{\beta}$ is a type $\beta$ simplex code over $\mathbb{Z}_{8}$ with parameters $(n=2^{2(k-1)}(2^{k}-1), M=8^k, d_{HW}=2^{2k-1}(2^{k}-1))$ generated by 
\begin{equation}
G_2^{\beta} = \left[\begin{array}{c|c|c|c|c} 11111111 & 0 & 2 & 4 & 6\\\hline
01234567 & 1 & 1 & 1 & 1 
\end{array}\right]
\end{equation}
and for $k>2$
\begin{equation}\label{skbeta}
G_k^{\beta} = \left[\begin{array}{c|c|c|c|c} 11 \cdots 1 & 00 \cdots 0 &
22 \cdots 2 & 44 \cdots 4 &66 \cdots 6\\\hline
G_{k-1}^{\alpha} & G_{k-1}^{\beta} & G_{k-1}^{\beta}&G_{k-1}^{\beta}& G_{k-1}^{\beta} 
\end{array}\right],
\end{equation}
where $G_{k-1}^{\alpha}$ is the generator matrix of $S_{k-1}^{\alpha}$.  The dual code of $\S_{k}^{\beta}$ is denoted by ${S_k^{\beta}}^{\perp}$.

\begin{theorem}
$r_{HW}(S_{k}^{\alpha})\geq 2^{3k+1}$ and $r_{E}(S_{k}^{\alpha})\leq 6(8^{k}-1)+2$.
\end{theorem}
\begin{proof}
The proof can be obtained using Proposition \ref{mattson}, Theorem ~\ref{brep_e_1to7}, equation (\ref{skalpha}) and is similar to $\mathbb{Z}_{4}$ case \cite{guptadurai14}. Hence omitted.

 \qed
\end{proof}

\begin{theorem}
$r_{E}(S_{k}^{\beta})\leq \frac{3}{2}(8^{k}-1)-\frac{5}{3}(4^{k}-1)-\frac{39}{2}+r_{E}(S_2^{\beta})$ and
$r_{HW}(S_{k}^{\beta}) \leq 3 (8^k-1)-\frac{10}{3} (4^k-1) -139 +r_{HW}(S_2^{\beta})$.
\end{theorem}
\begin{proof} First inequality is proved using Theorem ~\ref{brep_e_1to7} and is similar to $\mathbb{Z}_{4}$ case \cite{guptadurai14}. 
The case of homogeneous weight is similar.
 \qed
\end{proof}
\begin{theorem}
$r_E({S_k^{\alpha}}^{\perp})\leq 3$, $r_{HW}({S_k^{\alpha}}^{\perp}) = 1$ and $r_{HW}({S_k^{\beta}}^{\perp}) = 2$.
\end{theorem}
\begin{proof}
By Lemmaa 4.2 and Theorem 4.3 of \cite{gupta}, $r_E({S_k^{\alpha}}^{\perp})\leq 3$. By Theorem 4.3(3) of \cite{gupta}, $r_{HW}({S_k^{\alpha}}^{\perp})\leq 1$. Sience $r_{HW}({S_k^{\alpha}}^{\perp})\geq 1$, so $r_{HW}({S_k^{\alpha}}^{\perp}) = 1$. By Theorem 4.4 of \cite{gupta} and by Theorem \ref{delsarte},  $r_{HW}({S_k^{\beta}}^{\perp})\leq 2$ and as  $r_{HW}({S_k^{\beta}}^{\perp})\geq 1$ thus $r_{HW}({S_k^{\beta}}^{\perp}) = 1 \mbox{~or~} 2$ but $r_{HW}({S_k^{\beta}}^{\perp}) \neq 1$ by Proposition \ref{sphere}. Hence the result follows.\qed
\end{proof}
\begin{theorem}
$S_k^{\alpha}$ and $S_k^{\beta}$ are self orthogonal codes over $\mathbb{Z}_{8}$.
\end{theorem}
\begin{proof}
Proof follows from Proposition \ref{self_orthogonal_cond}.
\end{proof}

\section{Octonary MacDonald Codes of Type $\alpha$ and $\beta$}

The $q$-ary MacDonald code $\mathbb{M}_{k,u}(q)$
over the finite field $\mathbb{F}_q$ is a unique 
$[\frac{q^{k}-q^{u}}{q-1}, k,\\ q^{k-1}-q^{u-1}]$ 
code in which every nonzero codeword has
weight either $q^{k-1}$ or $q^{k-1}-q^{u-1}$ \cite{dodu}.
In \cite{colbourn}, authors have defined  the MacDonald codes over $\mathbb{Z}_4$ using the generator matrices of
simplex codes. In a similar manner one can define MacDonald code over $\mathbb{Z}_{2^s}$. For $1 \leq u \leq k-1,$ let
$G_{k,u}^{\alpha}\left(G_{k,u}^{\beta}\right)$ be the matrix
obtained from $G_{k}^{\alpha}\left(G_{k}^{\beta}\right)$ by
deleting columns corresponding to the columns of
$G_{u}^{\alpha}\left(G_{u}^{\beta}\right)$. i.e, 
\begin{equation}
\label{macalpha}
G_{k,u}^{\alpha}=\left[\begin{array}{cc}G_k^{\alpha}& \backslash\;
\frac{\bf{0}}{G_u^{\alpha}}
\end{array} \right],
\end{equation}
and\\
\begin{equation} \label{macbeta}
G_{k,u}^{\beta}=\left[\begin{array}{cc}G_k^{\beta}&
\backslash\; \frac{\bf{0}}{G_u^{\beta}}
\end{array} \right],
\end{equation}
where $[A \backslash \frac{\bf {0}}{B}]$ is the matrix obtained by deleting the matrix ${\bf 0}$ and $B$ from $A$ where
 $B$ is a $(k-u) \times 2^{su}$ matrix in $(\ref{macalpha})$ and $\left(\;\mbox{resp.}\;(k-u)
\times 2^{(s-1)(u-1)}(2^{u}-1)\right)$ matrix in $(\ref{macbeta})$ .
The code 
\[
\mathbb{M}_{k,u}^{\alpha}:[2^{sk}-2^{su},sk] \left(\mathbb{M}_{k,u}^{\beta}: [2^{(s-1)(u-1)}(2^{k}-1)-2^{(s-1)(u-1)}(2^{u}-1),sk]\right)
\]
generated by the matrix
$G_{k,u}^{\alpha}\left(G_{k,u}^{\beta}\right)$ is the punctured
code of $S_k^{\alpha}\left(S_k^{\beta}\right)$ and is  called a
{\em MacDonald code} of type $\alpha\; ( \beta)$.

Next theorem provides basic bounds on the covering radii of MacDonald codes over $\mathbb{Z}_8$.

\begin{theorem}
\[
\begin{array}{ccc}
r_E(\mathbb{M}_{k,u}^{\alpha}) & \leq & 6(8^k-8^r) + r_E(\mathbb{M}_{r,u}^{\alpha})\;\mbox{for}\; u < r \leq k,\\
\end{array}
\]
\end{theorem}

\begin{proof} Similar to $\mathbb{Z}_{4}$ case \cite{guptadurai14}.
\end{proof}
\qed

\section{Octonary Reed-Muller Code}
In this section we give covering radius of octonary first order Reed Muller code \cite{gupta}.
Let $1 \leq i \leq m-2.$ Let $\v_i$ be a vector of length $2^{m-2}$
consisting of successive blocks of $0$'s and $1$'s each of size $2^{(m-2)-i}$
and let ${\bf 1}=(111 \ldots 11) \in \;\Z_{2}^{\;2^{m-2}}.$ Let $G$ be a
$(m-1)\times 2^{m-2}$ matrix given by (consisting of the rows as ${\bf 1}$
and $4 \v_i \; (1 \leq i \leq m-2)$)
\begin{equation}\label{frmg}
G=\left[\begin{array}{ccccccccccc}
0&0&\cdots&0&0&4&4&\cdots&4&4\\
\vdots&\vdots&\ddots&\vdots&\vdots&\vdots&\vdots&\ddots&\vdots&\vdots\\
0&4&\cdots&0&4&0&4&\cdots&0&4\\
1&1&\cdots&1&1&1&1&\cdots&1&1\\
\end{array}\right].
\end{equation}
The code generated by $G$ is called the {\em first order Reed-Muller code over
$\;\Z_{8}$}, denoted $\mathbb{R}^{1,m-2}$. It is a $\left(n=2^{m-2}, M=2^{m+1}, d_{HW}=2^{m-1}\right)$ type $\alpha$ linear code over $\;\Z_{8}$ \cite{gupta}.
From Proposition \ref{HW_H} and \cite{helleseth} we have 
\begin{theorem}
$r_{HW}(\mathbb{R}^{1,m-2})= 2^{m-1}-2^{\frac{m}{2}-1} \mbox{for even} ~m$.
\end{theorem}

\section{Octonary Octa Codes}
The octa code over $\mathbb{Z}_{8}$ is generated by the following matrix.

\[
\mathbb{G} = \left[
\begin{array}{cccccccc}
5 & 7 & 5 & 6 & 1 & 0 & 0 & 0 \\
5 & 0 & 7 & 5 & 6 & 1 & 0 & 0 \\
5 & 0 & 0 & 7 & 5 & 6 & 1 & 0  \\
5 & 0 & 0 & 0 & 7 & 5 & 6 & 1
\end{array} \right]
\]

By Proposition \ref{sphere} we get the following result.
\begin{theorem}
If ${\cal C}$ is the code generated by $G$ then $r_{HW}({\cal C})\geq 6$.
\end{theorem}

\section{Conclusion}
In this work,  we have introduced new torsion and reduction codes for any linear octonary code and obtained a nice relationship among various reduction and torsion codes. Further we have extended 
some of the results regarding covering radius of \cite{guptadurai14} to the octonary case. In particular, we have found exact values and bounds of the covering radius of Repetition codes, Simplex codes of Type $\alpha$ and Type $\beta$ and their duals, MacDonald codes, and first order Reed-Muller codes, Octacode over $\mathbb{Z}_8$. New Reduction and torsion codes can be used to classify octonary linear codes. 


%
%
%
%
%
\end{document}